\documentclass{article}

\usepackage[top=2.5cm, bottom=2.5cm, left=2.5cm, right=2.5cm]{geometry}

\usepackage{anyfontsize}
\usepackage{pgfplots}
\pgfplotsset{width=7cm,compat=1.9}
\usepackage{tikz}%,tikz-3dplot}
\usetikzlibrary{arrows}
\usetikzlibrary{shapes}
\usetikzlibrary{arrows,positioning,shapes,fit,calc}
\pgfdeclarelayer{background}
\pgfsetlayers{background,main}

\usepackage{graphicx}
\usepackage{amsmath}
\usepackage{amsthm}
\usepackage{amsfonts}
\usepackage{stmaryrd}
\usepackage{color}
\usepackage{mathtools}
\DeclarePairedDelimiter\ceil{\lceil}{\rceil}
\DeclarePairedDelimiter\floor{\lfloor}{\rfloor}

\newcommand{\eqclass}[2]{[{#1}]_{{#2}}}
\newcommand{\pow}[2]{{#1} \times 10^{{#2}}}
\newcommand{\ourparagraph}[1]{{\smallskip\noindent {\bf {#1}.}}}

\newcommand{\Z}{\mathbb{Z}}

\newcommand{\lr}[2]{\llbracket {#1}; {#2} \rrbracket}

\newcommand{\bT}{\mathbb{T}}
\newcommand{\bS}{\mathbb{S}}

\renewcommand\vec{\mathbf}
\newcommand{\xA}{\vec{x_\mathbb{A}}}
\newcommand{\xT}{\vec{x_\mathbb{T}}}
\newcommand{\xS}{\vec{x_\mathbb{S}}}

\newcommand{\XT}{X_\mathbb{T}}
\newcommand{\XS}{X_\mathbb{S}}

\newcommand{\pT}{\pi_{\bT}}
\newcommand{\pS}{\pi_{\bS}}

\newcommand{\cX}{\mathcal{X}}
\newcommand{\cY}{\mathcal{Y}}
\newcommand{\cZ}{\mathcal{Z}}

\newcommand{\pY}{\pi_{Y}}
\newcommand{\pZ}{\pi_{Z}}

\DeclareMathOperator{\HH}{H}
\DeclareMathOperator{\V}{V}

\newtheorem{lemma}{Lemma}
\newtheorem{theorem}{Theorem}
\newtheorem{corollary}{Corollary}
\newtheorem{recall}{Recall}
\newtheorem{definition}{Definition}
\newtheorem{example}{Example}
\newtheorem{assumption}{Assumption}

\begin{document}

%\title{Scalable Information-Flow Analysis of \\Secure Three-Party Affine Computations}
%\date{}
%\author{}
%
%\maketitle

%%%%%%%%%%%%%%%%%%%%%%%%%%%%%%%%%%%%%%%%%%
\begin{center}
{\Large \bf Scalable Information-Flow Analysis of \\Secure Three-Party Affine Computations}\\

\bigskip
\bigskip
{\small Patrick Ah-Fat and Michael Huth\\
Department of Computing, Imperial College London\\
London, SW7 2AZ, United Kingdom\\
$\{$patrick.ah-fat14, m.huth$\}$@imperial.ac.uk}
\end{center}

\date{\today}

\bigskip

\begin{abstract}
Elaborate protocols in Secure Multi-party Computation enable several participants to compute a public function of their own private inputs while ensuring that no undesired information leaks about the private inputs, and without resorting to any trusted third party. However, the public output of the computation inevitably leaks \emph{some} information about the private inputs. 
Recent works have introduced a framework and proposed some techniques for quantifying such information flow. 
Yet, owing to their complexity, those methods do not scale to practical situations that may involve large input spaces. 
The main contribution of the work reported here is to formally investigate the information flow captured by the min-entropy in the particular case of secure three-party computations of affine functions in order to make its quantification scalable to realistic scenarios. 
To this end, we mathematically derive an explicit formula for this entropy under uniform prior beliefs about the inputs. We show that this closed-form expression can be computed in time constant in the inputs sizes and logarithmic in the coefficients of the affine function. Finally, we formulate some theoretical bounds for this privacy leak in the presence of non-uniform prior beliefs. 
\end{abstract}


\section{Introduction}

Secure Multi-party Computation (SMC) is a domain of cryptography that aims at enabling several parties to compute a public function of their own private inputs, while keeping the inputs secret and without resorting to any trusted third party \cite{yao1986generate,yao1982protocols,shamir1979share,%
rabin1989verifiable,ben1988completeness,chaum1988multiparty}. 
Multi-party secure protocols typically require the parties to engage in a series of rounds of communication in order to exchange their information so as to be able to collaboratively compute the intended output. Such protocols provide the guarantee that none of the parties will be able to infer any information about the other parties' input, other than the information conveyed by the public output itself. 
%no information about each of the private inputs leaks to the other parties, other than the information conveyed by the public output itself. 

Paradoxically, as a function of the inputs, the public output inevitably leaks some information about those private inputs. 
This leakage is considered as an inherent consequence of the primary objective of SMC: it is commonly qualified as the ``acceptable leakage'' and its study is thus largely ignored in the SMC literature \cite{lindell2009secure,orlandi2011multiparty,%
cramer2015secure,aumann2007security}. 
Recent works have been undertaken with the aim of quantifying such information flows \cite{ah2017secure,ah2018optimal,8573818}. 
By adapting techniques from Quantitative Information Flow (QIF) and applying concepts from Information Theory (IT) to the context of SMC, they introduce an attack model and a general notion of entropy that enable us not only to reason about the acceptable leakage in SMC, but also to construct bespoke privacy-enhancing mechanisms aimed at protecting the inputs' secrecy. 
In this attack model, the entropy of a targeted input reflects the amount of information that is gained by an attacker once the output is revealed. 

Although these techniques offer a rich framework designed for analysing information flows in SMC, their computation is essentially combinatorial, and their application in practice is thus impeded by the scalability of computing this combinatorics. 
%these techniques are not scalable to the computations that involve large input spaces. 
Indeed, in the general case, the time complexity of computing such entropy measures is quadratic in the product of the inputs sizes, making them inadequate for examining real world applications of SMC that may involve large input spaces. 
We believe however, that developing techniques that can perform such analyses efficiently would benefit and complement the extensive researches  \cite{kolesnikov2008improved,lindell2015efficient,lindell2012secure,%
araki2016high} 
that are being conducted on efficient SMC protocols: potential participants of an SMC would not only have efficient cryptographic protocols at their disposal, but they could also effectively run privacy analyses in order to precisely estimate the risk that they would run by entering the computation. 

In this paper, our objective is to focus our efforts on a particular class of functions for which we further investigate those analyses in order to make them applicable to arbitrarily large input spaces. More precisely, we focus on secure three-party computations, and we study the class of functions that are affine in the target's and the spectator's inputs, while the amount of information that an attacker gains on a targeted input will be measured by means of conditional min-entropy. 
In this setting, the main contribution of this work is to reduce the combinatorial essence of this information measure to a closed-form expression that has time complexity constant in the inputs sizes, and logarithmic in the coefficients of the affine function. 
More specifically, we show that under uniform prior beliefs, the conditional min-entropy can be reduced to a simple function of the size of the output domain, for which we then derive an explicit expression. 
Finally, as this reduction is valid under uniform prior beliefs on the inputs, we also exhibit some explicit bounds for this information measure in the presence of non-uniform prior beliefs. 

\ourparagraph{Outline of Paper}
We present an intuitive overview of our main contributions and of the key technical aspects of our work in Section \ref{sec:method}. 
We discuss some related works in Section \ref{sec:related}. 
The mathematical formalisation required for analysing information flows in secure three-party affine computations is introduced in Section \ref{sec:if_affine}. In Section \ref{sec:uni}, we show that the information gained by an attacker under uniform prior beliefs is entirely determined by the size of the output domain, for which we derive a closed-form expression. Explicit bounds for the information flow under non-uniform prior beliefs are presented in Section \ref{sec:non_uni}. We illustrate those theoretical results in Section \ref{sec:examples} and Section \ref{sec:conclu} concludes the paper.

\ourparagraph{Notations} 
%First, let us introduce some notations %, relative to discrete probability distributions, 
%that we will need in this work. 
Let $D$ be a discrete set. We denote by $|D|$ the cardinality of set $D$.
Let $\Omega(D)$ be the set of all probability distributions whose support is contained in $D$. 
Throughout, we present distributions as Python dictionaries with domain values as keys and associated probabilities as values. For example, $\{4\colon\frac{1}{2}, 8\colon\frac{1}{2}\}$ represents the uniform distribution over $\{4, 8\}$. 
For any integers $a$ and $b$, we will write $\lr{a}{b}$ for the set of consecutive integers ranging from $a$ to $b$, namely $\{a, a+1, \cdots, b\}$. 
The greatest common divisor of $a$ and $b$ will be denoted as $\gcd(a, b)$. 
The fact that two integers $i$ and $j$ have same residue modulo another integer $k$ will be denoted as $i = j \mod k$. 
%The set of positive integers will be denoted by $\Np$ while $\Rp$ and $\Rpe$ will denote the set of non-negative and positive real numbers respectively. 
%Let $n$ be in $\Np$. 
%A \emph{linear distribution} over $\llrr{1}{n}$ will refer to the triangular distribution with mode $n$, i.e. to the distribution
%$\{k \colon \frac{2k}{n(n+1)} \mid 1 \leq k \leq n\}$
%where $\frac{2}{n(n+1)}$ is a normalising factor. 
Given random variable $X$ and value $x$, the event ``$X = x$'' will be abbreviated by ``$x$'' when there is no ambiguity, and its probability will be denoted by $p(x)$. 
Similarly, we will abbreviate $\sum_{x \in D}$ by $\sum_x$ when the domain $D$ is obvious from context. 
Finally, the logarithm in base $2$ will be denoted as $\log$.

\section{Methodology}
\label{sec:method}

In this section, we present an overview of our main contributions and we highlight the key technical components of our work intuitively. 
Although the aim of this section is to illustrate and summarise our results, the detailed and rigorous approach is developed in Sections \ref{sec:if_affine}, \ref{sec:uni} and \ref{sec:non_uni}. 
This work is motivated by Secure Multi-party Computation, which requires all the manipulated values to belong to finite spaces. Thus, we will focus on integer values ranged in finite intervals. 

The main contribution of this work is to introduce an efficient and scalable way of quantifying the acceptable leakage in three-party affine computations. 
To this effect, we consider the secure computation of a public function $f$ performed on three private inputs $x$, $y$ and $z$. 
We wish to quantify the amount of information that an attacker, who has control of or is being able to eavesdrop on the value of $x$, would gain on input $y$ once the output of $f$ is revealed. 
We focus on the functions $f$ whose output $o$ can, once the input $x$ controlled by the attacker is fixed, be expressed as a function of $y$ and $z$, in its simplest form, as $o = f(y, z) = \beta y + \gamma z$ where $\beta$ and $\gamma$ are constant integers. 
We quantify the information gained by the attacker from this computation by $\HH(Y \mid O)$, the min-entropy of input $y$ given output $o$, considered as random variables. 
When inputs $Y$ and $Z$ are considered as random variables uniformly distributed on some intervals, we show in Section \ref{subsec:reduce} that this entropy can be reduced to an explicit formula involving  $N_O$, the number of possible values that output $O$ can take. 
The main difficulty now resides in deriving a closed-form formula for $N_O$, the focus of Section \ref{subsec:output_domain}, and for which we sketch an intuitive explanation now. 
Given that $Y$ and $Z$ are uniformly distributed on respective intervals $I_Y$ and $I_Z$, we show that those intervals can be assumed to be of the form $I_Y = \lr{0}{n}$ and $I_Z = \lr{0}{m}$ respectively, where $n$ and $m$ are positive integers. 
We also show that a simple simplification enables us to assume that constants $\beta$ and $\gamma$ are positive and coprime. 
The number of outputs $N_O$ can now be expressed as the following cardinal $N_O = |A|$ where we define set $A$ as: 

\[ A =  \{ \beta y + \gamma z \mid (y, z) \in \lr{0}{n} \times \lr{0}{m} \} \]

For the sake of our later explanation, we will define for all $i$ in $\lr{0}{n}$, the set $A_i$ as:

\[ A_i = \{ \beta i + \gamma z \mid z \in \lr{0}{m} \} \]
so that set $A$ can now be expressed as: 
\[A = \bigcup_{i = 0}^n A_i \]

In order to illustrate our method and theorems for computing the cardinal $|A|$, we will construct some graphical representations of set $A$ under different configurations in the following examples. 

\begin{example}
\begin{figure}
\centering
\includegraphics[scale=0.348]{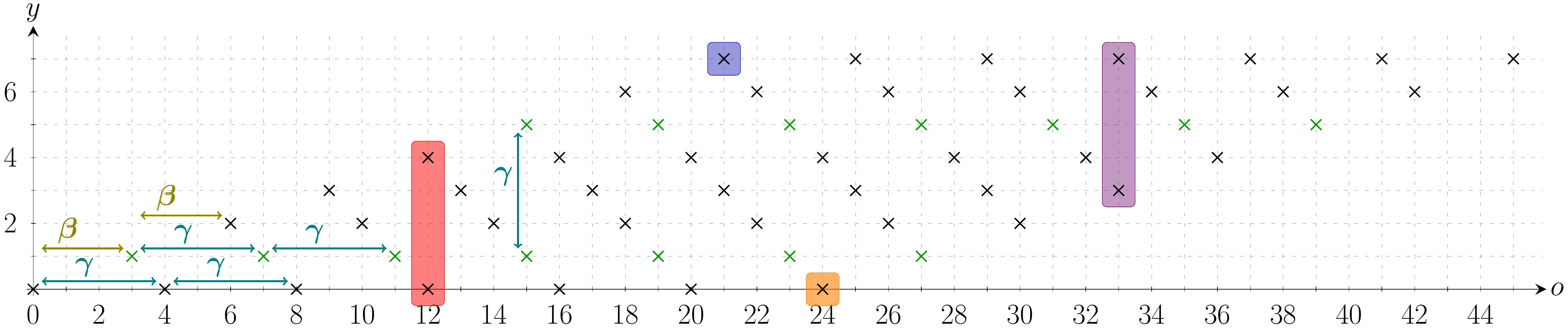}
\caption{Graphical representation of set $A = \{ \beta y + \gamma z \mid (y, z) \in \lr{0}{n} \times \lr{0}{m} \}$ of feasible outputs under parameters $\beta = 3, \gamma = 4, n = 7, m = 6$. }
\label{fig:ex_inter}
\end{figure}

Let $\beta = 3$, $\gamma = 4$, $n = 7$ and $m = 6$. The graphical representation of the corresponding set $A$ is shown in Figure \ref{fig:ex_inter}. The $y$-axis corresponds to the values of $y$, which is ranged in $\lr{0}{7}$, while the $x$-axis corresponds to the possible values that the output can take. 
For each row, indexed by $i \in \lr{0}{7}$, we mark by a cross the possible values that the output $o = \beta i + \gamma z$ can take. In other words, each row $i$ will represent the elements contained in set $A_i$. As $A$ is defined as the union of all the $A_i$, the set $A$ corresponds to the projection of all the crosses onto the $x$-axis. In other words, value $o$ belongs to set $A$ if and only if there is at least one cross in column $o$. 
\end{example}

In order to tally the number of feasible outputs, we will highlight some intuitive results, which we will formalise and prove in Section \ref{subsec:output_domain}. 
We first notice that the number of outputs is upper bounded by $(n+1)(m+1)$, and may be strictly lower than this bound since one column may contain several crosses. We will refer to such columns containing more than $1$ cross as \emph{intersections}. 
We make the following observations:
\begin{enumerate}
\item 
\label{obs:first}
The ``first'' intersection occurs at column $\beta \gamma = 12$ and is highlighted in red in Figure \ref{fig:ex_inter}. In other words, the lowest output $o$ whose column contains at least two crosses is $o=\beta \gamma$. 
\item We indicate in blue the first cross of the last row, indexed at column $\beta n = 21$, and in orange the last cross of the first row, indexed at column $\gamma m = 24$. We notice that the first intersection can occur if and only if both of those crosses, highlighted in blue and orange, do not stand before the column highlighted in red. 
In other words, the first intersection can occur if and only if $\beta n \geq \beta \gamma$ and $\gamma m \geq \beta \gamma$, i.e. if $n \geq \gamma$ and $m \geq \beta$, which we claim in Lemma \ref{lmm:single}. 
If one of those conditions is not satisfied, then there is no intersection, and the number of outputs is $(n+1)(m+1)$, which we claim in Corollary \ref{crllr:injective} and illustrate in the next example. 
\item Set $A$ is ``symmetrical'', i.e. that for all output $o$ in $\lr{0}{\beta n + \gamma m}$, we have:
\[ o \in A \iff \beta n + \gamma m - o \in A \]
where $\beta n + \gamma m$ is the largest output obtained for maximal values of $y$ and $z$. This is proved in Lemma \ref{lmm:symmetric}. 
\item Thus, the last intersection occurs at column $\beta n + \gamma m - \beta \gamma$ and is highlighted in purple. 
Together with observation \ref{obs:first}, this constitutes the content of Lemma \ref{lmm:range_double}. 
Moreover, by symmetry, there is the same number of outputs contained in $\lr{0}{\beta \gamma - 1}$ and $\lr{\beta n + \gamma m - \beta \gamma + 1}{\beta n + \gamma m}$, as claimed in Corollary \ref{crllr:count_last}. 
\item As there is no intersection before the red column $\beta \gamma$, the number of outputs contained in $\lr{0}{\beta \gamma - 1}$ can be obtained by summing the number of elements of all $A_i$ contained in this interval. More formally, we have:
\[ |A \cap \lr{0}{\beta \gamma - 1}| = \sum_{i=0}^n |A_i \cap \lr{0}{\beta \gamma - 1}| \]
This corresponds to the total number of crosses that stand before the column highlighted in red. We develop its computation in Theorem \ref{thm:count_first}. 
\item Finally, we observe that all the columns lying between the red and purple ones, i.e. ranging in the interval $\lr{\beta \gamma}{\beta n + \gamma m - \beta \gamma}$, contain at least one cross. This result is formalised in Theorem \ref{thm:middle_inc} and implies that:
\[ |A \cap \lr{\beta \gamma}{\beta n + \gamma m - \beta \gamma}| = \beta n + \gamma m - 2 \beta \gamma + 1 \]

In order to prove that all such columns contain at least one cross, we make the following reasoning. 
\begin{enumerate}
\item Two sets $A_i$ whose indices are separated by a multiple of $\gamma$ will only contain some outputs that have the same residue modulo $\gamma$. 
More precisely, if $i$ and $j$ are both congruent to some $k$ modulo $\gamma$, then the elements of $A_i$ and $A_j$ will be congruent to $\beta k$ modulo $\gamma$. 
We illustrate this fact in Figure \ref{fig:ex_inter}, where we color in green the elements of sets $A_1$ and $A_5$. We can notice that all those elements are congruent to $1 \beta$ modulo $\gamma$. 
\item For all $k$ in $\lr{0}{\gamma  - 1}$, let us define $B_k$ as the union of all the $A_i$ whose index $i$ is congruent to $k$ modulo $\gamma$:
\[ B_k = \bigcup_{i = k \mod \gamma} A_i \]

For example, $B_1$ can be represented as the projection of all the green crosses on the $x$-axis. 
Then, we can see that each $B_k$ includes all the outputs that are ranged between the red and purple columns and that are congruent to $\beta k$ modulo $\gamma$. 
This observation is formalised in Lemmas \ref{lmm:bj} and \ref{lmm:bj_inc}. 
We can indeed see in the figure that $\{15, 19, 23, 27, 31\} \subseteq A_1 \cup A_5 = B_1$. We notice that this may not be the case outside of the domain delimited by the red and purple columns as for example $47 \notin A_1 \cup A_5$. 
\item Finally, as formally explained in Theorem \ref{thm:middle_inc}, we claim that for all output $o$ ranged between the red and purple columns, there exists a $k$ in $\lr{0}{\gamma - 1}$ such that $o \in B_k$. 
Indeed, if we denote by $r$ the residue of $o$ modulo $\gamma$, it suffices to choose $k = \beta^{-1} r$ to ensure that $\beta k = r \mod \gamma$, which then implies $o \in B_k$. 
While $\beta^{-1}$ refers to the inverse of $\beta$ modulo $\gamma$, such an operation is allowed since $\beta$ and $\gamma$ are coprime. 
For example, output $o=17$ has residue $r=1$ modulo $\gamma=4$. 
In this case, we can choose $k = \beta^{-1}r = 3^{-1}\cdot 1 = 3$ and we can verify that $o \in B_3$. 
This concludes our intuition and ensures that all column ranged between the red and purple one will contain at least one cross. 
\end{enumerate}
\end{enumerate}

Finally, Theorem \ref{thm:lower_bound} and Corollary \ref{crll:lower_z} derive some lower and upper bounds for $\HH(Y \mid O)$ when prior beliefs on the inputs are not uniform. 

\begin{example}
\begin{figure}
\centering
\includegraphics[scale=0.348]{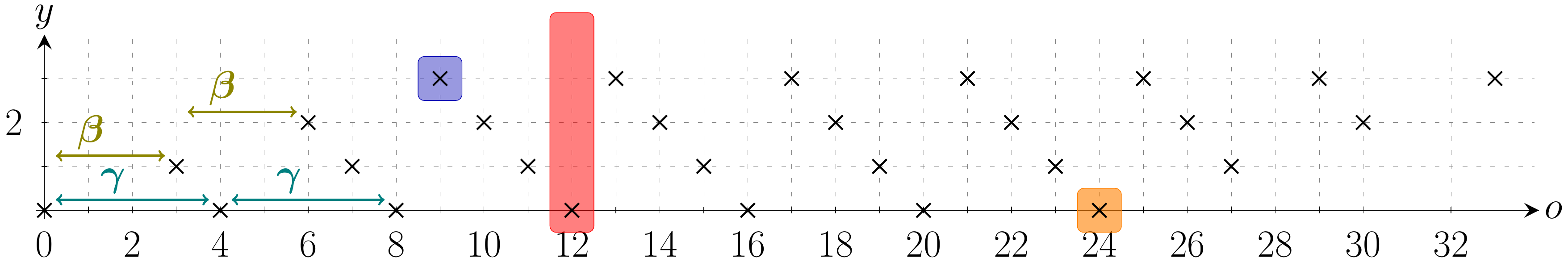}
\caption{Graphical representation of set $A = \{ \beta y + \gamma z \mid (y, z) \in \lr{0}{n} \times \lr{0}{m} \}$ of feasible outputs under parameters $\beta = 3, \gamma = 4, n = 3, m = 6$. }
\label{fig:ex_inter_bis}
\end{figure}

Let $\beta = 3$, $\gamma = 4$, $n = 3$ and $m = 6$. In the graphical representation of the corresponding set $A$ that is displayed in Figure \ref{fig:ex_inter_bis}, we can notice that the blue cell representing $\beta n$ appears before the red column $\beta \gamma$, meaning that the condition $n \geq \gamma \wedge m \geq \beta$ is not satisfied. 
This implies that there is no intersection in this setting and that $|A| = (n+1)(m+1) = 28$. 
\end{example}

\section{Related Works}
\label{sec:related}

In this section, we discuss related work that constitutes the foundations and the motivations of our present work. 

\subsection{Secure Multi-party Computation}

Secure Multi-party Computation  \cite{yao1986generate,yao1982protocols,shamir1979share,%
rabin1989verifiable,ben1988completeness,chaum1988multiparty} is a domain of cryptography that provides advanced protocols which enable several participants to compute a public function of their own private inputs without having to rely on any other trusted third party or any external authority. Those protocols enable the participants to compute a function in a decentralised manner, while ensuring that no information leaks about the private inputs, other than that which can be inferred from the public output. 
The commonly called ``acceptable leakage'' which is further studied in this paper, is the information that can be inferred from an attacker about the other inputs given the knowledge of the public output. 
Secure multi-party computation is not the only domain that is subject to an acceptable leakage. In particular, the results of our work are also applicable to other fields or scenarios that aim at protecting the inputs' privacy and that involve the opening of a public output, such as outsourced computation where a trusted third party is privately sent all the inputs and returns the public output as unique piece of information, or trusted computing where the parties input their secret data into hardware security modules, which then ensure that no unintended information will be accessible to the other parties.

\subsection{Quantitative Information Flow}

The purpose of Quantitative Information Flow (QIF) \cite{smith2009foundations,malacaria2015algebraic} is to provide frameworks and techniques based on information theory and probability theory for measuring the amount of information that leaks from a secret. 
Different mathematical concepts have emerged in order to convey varied and precise information about a secret: Shannon entropy \cite{BLTJ:BLTJ1338} reflects the minimum number of binary questions required to recover a secret on average, while the min-entropy is an indicator of the probability to guess a secret in one try \cite{smith2011quantifying,cachin1997entropy,smith2009foundations}. 
Richer measures such as R\'enyi entropy \cite{renyi1961measures} and the $g$-entropy \cite{m2012measuring} have been introduced in order to quantify some specific properties of a secret, and more general entropies have been proposed in order to unify those different concepts \cite{ah2018optimal,khouzani2016relative}. 
In this work, we will measure the information gained by an attacker by means of min-entropy, which is used extensively in cryptography in order to quantify the vulnerability of a secret.

\subsection{Differential Privacy}

Differential Privacy (DP) \cite{dwork2008differential,dwork2014algorithmic} formalises privacy concerns and introduces techniques that provide users of a database with the assurance that their personal details will not have a significant impact on the output of the queries performed on the database. More precisely, it proposes mechanisms which ensure that the outcome of the queries performed on two databases differing in at most one element will be statistically indistinguishable. Moreover, minimising the distortion of the outcome of the queries while ensuring privacy is an important trade-off that governs DP. 
Although DP is particularly adapted for guaranteeing privacy in statistical computations involving a large number of parties, its effectiveness diminishes when a small number of parties are involved in the computation. 
For example, in a two-party computation, a DP mechanism would ensure that the output would not be sensibly affected when half of the data is changed. In this case, the utility of the computed function is thus be drastically hindered by the low number of parties. 
Unlike DP and other works that have been conducted on trading off privacy and utility in SMC, this work does not intend to enhance the inputs privacy. 
Instead, our objective is to propose an efficient method for quantifying the privacy risks that a certain kind of computations presents. 

\subsection{Information Flow in Secure Multi-party Computation}

Recent works \cite{ah2017secure,ah2018optimal,8573818} have adapted techniques stemming from QIF to the setting of SMC in order to propose a model that allows us to reason about the acceptable leakage. 
In this model, the set of parties willing to compute a public function $f$ is partitioned into three sets: a set of attackers, a set of targets and a set of spectators, holding the respective input vectors $\xA$, $\xT$ and $\xS$. 
The attackers are those parties willing to share the value of their inputs and to take advantage of the public output of the computation $f(\xA, \xT, \xS)$ in order to learn as much information as possible on their targets' inputs, while the remaining parties are called spectators. 
From the point of view of the attackers, the inputs $\xT$ and $\xS$ are unknown values and are thus modelled as random variables $\XT$ and $\XS$, further deemed to be independent since targets and spectators are supposed to be honest parties who provide their inputs without being influenced by any other information. The attackers' prior belief on those inputs will represent the prior distributions $\pT$ and $\pS$ of those random variables. 
The output of the function $f$ is then also considered as a random variable defined as $O = f(\xA, \XT, \XS)$. 
The privacy of the targeted parties is then expressed as the conditional entropy $\HH(\XT \mid \xA, O)$ of the targeted inputs given knowledge of the attackers' inputs and the conditional knowledge of the output. 
The choice of the entropy measure $\HH$ depends on the users' privacy concerns and is left general in \cite{ah2018optimal} to this end. 
In this work and for clarity purposes, we will choose to convey the inputs' privacy by means of min-entropy, as in many cryptographic scenarios, although our analyses can be adapted to more general entropy measures. 
Under this assumption, the privacy of the targeted parties becomes:
\begin{eqnarray}
\HH(\XT \mid \xA, O) &=& - \log \sum_o p(o \mid \xA) \cdot \max_\xT p(\xT \mid \xA, o) \nonumber \\
&=& - \log \sum_o \max_\xT p(\xT) \cdot p(o \mid \xA, \xT) \label{eq:awae_naive}
\end{eqnarray}
by virtue of Bayes theorem. Moreover, as $\XT$ and $\XS$ are independent, we know that: 
\[ p(o \mid \xA, \xT) = \sum_{\substack{\xS \\ f(\xA, \xT, \xS) = o}} p(\xS) \]

If we denote by $n$ and $m$ the size of the domains of $\XT$ and $\XS$ respectively, we know that computing one $p(o \mid \xA, \xT)$ has complexity $\mathcal{O}(m)$ and thus computing each $\max_\xT p(\xT) \cdot p(o \mid \xA, \xT)$ has complexity $\mathcal{O}(nm)$. 
Moreover, in the worst case, i.e. if $f$ is injective, the output domain will have a size of $nm$, which yields an overall complexity in $\mathcal{O}(n^2 m^2)$ for the computation of $\HH(\XT \mid \xA, O)$. 
In conclusion, although recent works have introduced a framework for characterising and quantifying the acceptable leakage, its computation cost is quadratic in the product of the inputs sizes in general, which prevents those privacy analyses to be applicable in practice, and this major complexity issue constitutes the focus of this paper.

\section{Information Flow Analysis in Secure Three-Party Affine Computations}
\label{sec:if_affine}

%	\subsection{Formal Setting}

Let us consider three parties $\cX$, $\cY$ and $\cZ$ holding the respective private inputs $x$, $y$ and $z$. 
Let $f$ be a public function of three variables. 
We assume that the parties wish to enter the secure computation of $f(x, y, z)$ and that $\cX$ is attacking $\cY$ under spectator $\cZ$. 
From the point of view of attacker $\cX$, although $x$ is a known and constant value, the inputs $y$ and $z$ appear as unknown values and will be modelled as random variables $Y$ and $Z$. 
Parties $\cY$ and $\cZ$ are supposed to be honest parties who will not collaborate. Thus, random variables $Y$ and $Z$ are deemed to be independent. 
We further assume that the target's and spectator's inputs are from finite intervals $I_Y$ and $I_Z$:
\[ Y \in I_Y , \quad Z \in I_Z \]

Their prior probability distributions $\pY$ and $\pZ$ will represent the prior beliefs that $\cX$ may have on those values, such that $\pY \in \Omega(I_Y)$ and $\pZ \in \Omega(I_Z)$. We note that the absence of prior belief may be represented as uniform prior distributions. 
Finally, we assume that function $f$ is affine in the target's and spectator's inputs, i.e. that we can choose three constant integers $\alpha$, $\beta$ and $\gamma$ so as to express the output of $f$ as:
\[ f(x, y, z) = \alpha + \beta y + \gamma z \]

Note that constants $\alpha$, $\beta$ and $\gamma$ may be function of input $x$, which is also considered as a constant. 
Admissible candidates for such affine functions $f$ can for example be defined as $f(x, y, z) = 3y + 4z$ or $f(x, y, z) = x^2 + xy + (x^3-2)z$. 

\begin{assumption}
From the attacker's point of view, input $x$ is a known value and will thus be considered as a constant throughout this paper.
\end{assumption}

Thus, we may  abuse notation by omitting the first argument of $f$, we refer to its output $o$ as:
\[ o = f(y, z) = \alpha + \beta y + \gamma z \]
while we define the corresponding random variable $O$ for the output as $O = f(Y, Z) = \alpha +  \beta Y + \gamma Z$. 
We also introduce the output domain $D_O$ as: 
\[D_O = \{ f(y, z) \mid (y, z) \in I_Y \times I_Z\}\]
%\[D_O = \{ f(y, z) \mid \exists (y, z) \in I_Y \times I_Z \colon p(y) > 0 \wedge p(z) > 0 \}\]

By denoting the min-entropy by $\HH$, the amount of information that the attacker gains on the targeted input once the output is revealed will be quantified by $\HH(Y \mid x, O)$. 
Since the value of $x$ will also be considered as a public constant in the present privacy analyses, we will refer to this quantity as $\HH(Y \mid O)$, which develops as:
\[ \HH(Y \mid O) = - \log \V(Y \mid O) \]
where the Bayes vulnerability of $Y$ given $O$ is defined as:
\begin{eqnarray}
\V(Y \mid O) &=& \sum_o p(o) \cdot \max_y p(y \mid o) \label{eq:vuln_dry} \\
&=& \sum_o \max_y p(y) \cdot p(o \mid y) \nonumber
\end{eqnarray}

To conclude this section and in order to simplify the following development, we will examine the particular case when $\beta$ or $\gamma$ is zero. 

\begin{lemma}
\begin{enumerate}
\item If $\beta = 0$ then $\HH(Y \mid O) = \HH(Y)$. 
\item If $\beta \neq 0$ and $\gamma = 0$ then $\HH(Y \mid O) = 0$. 
\end{enumerate}
\end{lemma}

\begin{proof}
\begin{enumerate}
\item
If $\beta = 0$ then clearly no information about $Y$ leaks from $O$ and thus $\HH(Y \mid O) = \HH(Y)$. More formally, in this case, $Y$ and $O$ are independent and thus Equation (\ref{eq:vuln_dry}) becomes:
\[
\begin{array}{lllll}
\V(Y \mid O) &=& \sum_o p(o) \cdot \max_y p(y) &=& 1 \cdot \HH(Y)
\end{array}
\]
\item 
If $\beta \neq 0$ and $\gamma = 0$ then $Y$ is entirely determined by $O$ given the relation $y = \frac{o - \alpha}{\beta}$ and thus $\HH(Y \mid O) = 0$. More formally, for all $o$ in $D_O$, there exists one $y$ in $I_Y$ such that $p(y \mid o) = 1$ and thus Equation (\ref{eq:vuln_dry}) becomes:
\[
\begin{array}{lllll}
\V(Y \mid O) &=& \sum_o p(o) \cdot 1 &=& 1
\end{array}
\]
\end{enumerate}
\end{proof}

\begin{assumption}
In the rest of the paper, we will assume that $\beta$ and $\gamma$ are non-zero. 
\end{assumption}

\section{Privacy under uniform prior beliefs}
\label{sec:uni}
	\subsection{Reducing the entropy expression}
	\label{subsec:reduce}
In this section, we study the case where the attacker has no prior belief on the target's and spectator's inputs, i.e. when $\pY$ and $\pZ$ are uniform on $I_Y$ and $I_Z$ respectively. In other words, we assume that for all $y$ in $I_Y$ and $z$ in $I_Z$, we have $p(y) = \frac{1}{|I_Y|}$ and $p(z) = \frac{1}{|I_Z|}$. As $\pY$ is uniform, we have:
\begin{eqnarray*}
\V(Y \mid O) &=& \sum_o \max_y p(y) \cdot p(o \mid y) \\
&=& \frac{1}{|I_Y|} \cdot \sum_o \max_y p(o \mid y)
\end{eqnarray*}

However, by definition, we know that for all output $o$ in $D_O$, there exists at least one pair $(y, z)$ in $I_Y \times I_Z$ that satisfies $f(y, z) = o$. For all such pairs, as $Y$ and $Z$ are independent, we have: 
\begin{eqnarray*}
p(o \mid y) &=& \sum_{\substack{z' \in I_Z \\ f(y, z') = o}} p(z') \\
&=& p(z)
\end{eqnarray*}
since for a given $o$ and $y$, there is at most one $z'$ that satisfies $f(y, z') = o$ as $f$ is affine and $\gamma$ is non-zero. Consequently, $p(o \mid y) = \frac{1}{|I_Z|}$ since $\pZ$ is uniform, and thus:
\begin{eqnarray}
\label{eq:vuln_uni}
\V(Y \mid O) &=& \frac{N_O}{|I_Y| \cdot |I_Z|}
\end{eqnarray}
where $N_O$ denotes the cardinal of $D_O$. 
Our aim will now be to compute $N_O$. 

We mention four simplifications before analysing this problem in more details. 

\begin{assumption}
\begin{enumerate}
\item
We first notice that deducting constant $\alpha$ from the output of $f$ does not affect the number of different outputs, which enables us to simplify $f$ as $f(y, z) = \beta y + \gamma z$. 
\item
Now, let us assume that interval $I_Y$ is of the form $I_Y = \lr{a}{b}$. 
By substituting variable $y$ to variable $y' = y - a$, we can rewrite the expression of $f$ as $f(y', z) = \beta a + \beta y' + \gamma z$. 
The new variable $y'$ is ranged in $\lr{0}{b-a}$ and we can again deduct constant $\beta a$ from the output. 
We can perform the same reasoning with variable $z$, which enables us to assume without loss of generality that inputs $y$ and $z$ belong to some intervals of the form $\lr{0}{n}$, $\lr{0}{m}$. 
\item 
We now show that integers $\beta$ and $\gamma$ can be assumed to be positive without loss of generality. 
If both $\beta$ and $\gamma$ are negative, then we can equivalently compute the number of outputs of function $f'(y, z) = - f(y, z)= -\beta y - \gamma z$ which has positive coefficients. 
If $\beta<0$ and $\gamma > 0$, we can write $f$ as $f(y, z) = \beta n - \beta (n -y) + \gamma z$. However, the input space $\lr{0}{n}$ of variable $Y$ is equal to that of $n - Y$, and we can thus equivalently study function $f'(y, z) = - \beta y + \gamma z$ whose coefficients are positive. 
Conversely, if $\beta > 0$ and $\gamma < 0$, we can again equivalently study $f'(y, z) = - f(y, z)$, which is tackled in the previous case. 
\item
Let us denote by $d$ the greatest common divisor of $\beta$ and $\gamma$. We know that $d$ can be computed in $\mathcal{O}(\log(\beta + \gamma))$. Function $f$ can be factorised as $f(y, z) = d \cdot f'(y, z)$ where $f'(y, z)= \beta' y + \gamma' z$ where $\beta'$ and $\gamma'$ are coprime. 
However, functions $f$ and $f'$ have the same number of outputs. 
We can thus assume that the coefficients of the affine function are coprime provided that we have computed their greatest common divisor. 
\end{enumerate}
\end{assumption}

	\subsection{Measuring the size of the output domain}
	\label{subsec:output_domain}

For the sake of clarity, this technical subsection will be developed so as to be self-contained. 

Let $n$, $m$ be two non-negative integers and $\beta$ and $\gamma$ be two positive integers. Let us also assume that $\beta$ and $\gamma$ are coprime. 
The aim is to calculate the cardinal $N_O$ of the set $A$ defined as follows:
\begin{equation*}
A = \{ \beta y + \gamma z \mid (y, z) \in \lr{0}{n} \times \lr{0}{m} \}
\end{equation*}

We can first notice that $|A|$ is positive and upper bounded by $(n+1)(m+1)$. 
The difficulty is that two different pairs $(y, z)$ and $(y', z')$ in $\lr{0}{n} \times \lr{0}{m}$ can satisfy $\beta y + \gamma z = \beta y' + \gamma z'$, and thus $|A|$ will often be lower than $(n+1)(m+1)$. 
We also notice that $A \subseteq \lr{0}{\beta n + \gamma m}$, and thus we also have $|A| \leq \beta n + \gamma m + 1$. 

\ourparagraph{Notations:} For any real $x$, the floor of $x$ will be denoted by $\floor{x}$ while $\ceil{x}$ will denote its ceiling. 
The fact that two integers $i$ and $j$ have same residue modulo another integer $k$ will be denoted as $i = j \mod k$. 
For all integers $i$ and $k$, we will denote the equivalence class of $i$ modulo $k$ by $\eqclass{i}{k} = \{ j \in \Z \mid j = i \mod k \}$. 

\begin{recall}
\label{rcll:floor}
For all real numbers $x$, we have 
\(
\ceil{-x} = - \floor{x}
\).
\end{recall}

\begin{proof}
Let $x$ be a real number. We have:
\[ -x \leq \ceil{-x} < -x +1 \]
and so:
\[ x - 1 < - \ceil{-x} \leq x \]
and thus, as $- \ceil{-x}$ is integral:
\[ - \ceil{-x} = \floor{x} \]
\end{proof}

\begin{recall}
\label{rcll:sum_floor}
Let $p$ and $q$ be two coprime natural numbers. We have:
\begin{equation*}
\sum_{k = 1}^{q-1} \floor{\frac{kp}{q}} = \frac{(p-1)(q-1)}{2}
\end{equation*}
\end{recall}

\begin{proof}
If $q = 1$ then the result is immediate since the sum adds up to $0$. 
Let us now assume that $q > 1$. 
We can notice that we have:
\begin{eqnarray*}
\sum_{k = 1}^{q-1} \floor{\frac{kp}{q}} 
&=& \sum_{k = 1}^{q-1} \floor{\frac{qp - kp}{q}} \\
&=& \sum_{k = 1}^{q-1} \floor{p - \frac{kp}{q}} \\
&=& (q-1)p + \sum_{k = 1}^{q-1} \floor{- \frac{kp}{q}}
\end{eqnarray*}
and thus:
\begin{equation}
\label{eq:sumfloor}
2 \sum_{k = 1}^{q-1} \floor{\frac{kp}{q}} 
= (q-1)p + \sum_{k = 1}^{q-1} 
\left( \floor{\frac{kp}{q}} + \floor{- \frac{kp}{q}} \right)
\end{equation}
However, because $p$ and $q$ are coprime we know that for all $k$ in $\lr{1}{q-1}$, we have $gcd(k, q) = 1$. 
As $q > 1$, we thus know that for all $k$ in $\lr{1}{q-1}$, we have $\frac{kq}{p} \notin \Z$ and thus:
\begin{eqnarray*}
\sum_{k = 1}^{q-1} \left( \floor{\frac{kp}{q}} + \floor{- \frac{kp}{q}} \right)
&=& \sum_{k = 1}^{q-1} (-1) \\
&=& (q-1)(-1)
\end{eqnarray*}
and thus Equation (\ref{eq:sumfloor}) becomes:
\begin{eqnarray*}
2 \sum_{k = 1}^{q-1} \floor{\frac{kp}{q}} 
&=& (q-1)p + (q-1)(-1) \\
&=& (q-1)(p-1)
\end{eqnarray*}
\end{proof}

\begin{recall}
\label{rcll:gengroup}
Let $p$ and $q$ be two coprime natural numbers. We have:
\begin{equation*}
\bigcup_{j \in \lr{0}{q - 1}} \eqclass{pj}{q} = \Z
\end{equation*}
\end{recall}

\begin{proof}
Let $i$ and $j$ be in $\lr{0}{q-1}$. We have:
\begin{eqnarray*}
\eqclass{pi}{q} = \eqclass{pj}{q}
&\iff& p(i-j) = 0 \mod q \\
&\iff& i = j \mod q
\end{eqnarray*}
since $p$ and $q$ are coprime. 
Thus, for all distinct $i$ and $j$ in $\lr{0}{q-1}$, we have $\eqclass{pi}{q} \neq \eqclass{pj}{q}$ and thus $|\{ pj \mod q \mid j \in \lr{0}{q-1} \}| = q$ and therefore:
\begin{eqnarray*}
\bigcup_{j \in \lr{0}{q - 1}} \eqclass{pj}{q} 
&=& \bigcup_{j \in \lr{0}{q - 1}} \eqclass{j}{q}  \\
&=& \Z
\end{eqnarray*}
\end{proof}

\begin{lemma}
\label{lmm:equi}
Let $(y, z)$ and $(y', z')$ be in $\lr{0}{n}\times \lr{0}{m}$. We have:
\begin{equation}
\beta y + \gamma z = \beta y' + \gamma z'
\iff
\exists k \in \Z \colon
\left(
y' = y + k \gamma
\wedge
z' = z - k \beta
\right)
\end{equation}
\end{lemma}

\begin{proof}
Let $(y, z)$ and $(y', z')$ be in $\lr{0}{n} \times \lr{0}{m}$. As $\beta$ and $\gamma$ are coprime, we have:
\begin{eqnarray*}
\beta y + \gamma z = \beta y' + \gamma z'
&\iff&
\beta (y' - y) = \gamma (z - z')
\\ &\iff&
\left(
\exists k \in \Z \colon 
y' - y = k \gamma
\right)
\wedge
\left(
\beta (y' - y) = \gamma (z - z')
\right)
\\ &\iff&
\exists k \in \Z \colon
\left(
y' = y + k \gamma
\wedge
z' = z - k \beta
\right)
\end{eqnarray*}
\end{proof}

\begin{lemma}
\label{lmm:single}
Let $(y, z)$ and $(y', z')$ be two distinct pairs in $\lr{0}{n} \times \lr{0}{m}$. 
We have:
\begin{equation}
\beta y + \gamma z = \beta y' + \gamma z'
\implies
n \geq \gamma \wedge m \geq \beta
\end{equation}
\end{lemma}

\begin{proof}
Let $(y, z)$ and $(y', z')$ be two distinct pairs in $\lr{0}{n} \times \lr{0}{m}$ such that:
\begin{equation*}
\beta y + \gamma z = \beta y' + \gamma z'
\end{equation*}

By virtue of Lemma \ref{lmm:equi}, we can take $k$ in $\Z$ such that:
\[
y' - y = k \gamma
\wedge
z - z' = k \beta
\]

As $(y, z)$ and $(y', z')$ are different, we further know that $k$ is different from $0$. 
This implies that:
\[
|y' - y| \geq \gamma \wedge |z'-z| \geq \beta
\]

But as $y$ and $y'$ belong to $\lr{0}{n}$ and $z$ and $z'$ belong to $\lr{0}{m}$, we also know that:
\[ |y' - y| \leq n \wedge |z'-z| \leq m \]
and thus:
\[ n \geq \gamma \wedge m \geq \beta \]
\end{proof}

\begin{corollary}
\label{crllr:injective}
If $n < \gamma \vee m < \beta$, then $|A| = (n+1)(m+1)$. 
\end{corollary}

\begin{proof}
Let us assume that $n < \gamma \vee m < \beta$. 
Let us define the function $g$ as \( g\colon (y, z) \longmapsto \beta y + \gamma z \) with domain $\lr{0}{n} \times \lr{0}{m} \longrightarrow \lr{0}{\beta n + \gamma m}$. 
By virtue of Lemma \ref{lmm:single}, we know that the function $g$ is injective. 
Thus, we have:
\[
|A| = |g(\lr{0}{n}\times \lr{0}{m})| = (n+1)(m+1)
\]
\end{proof}

\begin{assumption}
In the remainder of this section, we will now assume that $n \geq \gamma \wedge m > \beta$. 
\end{assumption}

\begin{lemma}
\label{lmm:range_double}
Let $o$ be in $A$. 
Let $(y, z)$ and $(y', z')$ be two distinct pairs in $\lr{0}{n}\times \lr{0}{m}$. 
We have:

\begin{equation}
\beta y + \gamma z = o \wedge \beta y' + \gamma z' = o
\implies
o \in \lr{\beta \gamma}{\beta n + \gamma m - \beta \gamma}
\end{equation}
\end{lemma}

\begin{proof}
Let $o$ be in $A$. 
Let $(y, z)$ and $(y', z')$ be two distinct pairs in $\lr{0}{n} \times \lr{0}{m}$ such that:
\[
\beta y + \gamma z = o \wedge \beta y' + \gamma z' = o
\]

By virtue of Lemma \ref{lmm:equi}, we can take $k$ in $\Z$ such that:
\begin{equation}
\label{eq:twice}
y' = y + k \gamma
\wedge
z' = z - k \beta
\end{equation}

We further know that both pairs are distinct and we can thus choose $k$ different from $0$. Without loss of generality, let us assume that $(y,z) >_2 (y', z')$ where $>_2$ refers to the lexicographic order on integer pairs. In other words, let us assume that $k>0$. 

We know that $z' \geq 0$ and Equation (\ref{eq:twice}) ensures that $y' \geq \gamma$ since $k>0$. 
As $o = \beta y' + \gamma z'$, we thus have $o \geq \beta \gamma$. 

Conversely, we know that $y' \leq n$ and Equation (\ref{eq:twice}) ensures that $z' \leq m - \beta$ since $k>0$. 
As $o = \beta y' + \gamma z'$, we thus have:
\begin{eqnarray*}
o &\leq& \beta n + \gamma (m - \beta) \\
&\leq& \beta n+ \gamma m - \beta \gamma
\end{eqnarray*}
\end{proof}

\begin{theorem}
\label{thm:count_first}
We have:
\begin{equation}
|A \cap \lr{0}{\beta \gamma - 1}| = \beta \gamma - \frac{(\beta - 1)(\gamma -1)}{2}
\end{equation}
\end{theorem}

\begin{proof}
By virtue of Lemma \ref{lmm:range_double}, we know that for all $o$ in $\lr{0}{\beta \gamma - 1}$, and for all pairs $(y, z)$ and $(y', z')$ in $\lr{0}{n}^2$, we have:
\[ \beta y + \gamma z = o \wedge \beta y' + \gamma z' = o 
\implies 
(y,z) = (y', z')
\]

Thus:
\begin{eqnarray*}
|A \cap \lr{0}{\beta \gamma - 1}| 
&=& | \{ (y, z) \in \lr{0}{n} \times \lr{0}{m} \mid \beta y + \gamma z \in \lr{0}{\beta \gamma - 1} \} | \\
&=& \sum_{y = 0}^{\gamma - 1} |\{ z \in \lr{0}{m} \mid \beta y + \gamma z \in \lr{0}{\beta \gamma - 1} \}|
\end{eqnarray*}
since $\beta y + \gamma z \geq \beta \gamma$ for all $z$ in $\lr{0}{n}$ when $y \geq \gamma$, and since $n \geq \gamma$. 
Moreover, since $m \geq \beta$, for all $y$ in $\lr{0}{\gamma - 1}$, we have:
\begin{eqnarray*}
|\{ z \in \lr{0}{m} \mid \beta y + \gamma z \in \lr{0}{\beta \gamma - 1} \}| 
&=& |\{ z \in \lr{0}{m} \mid \gamma z \in \lr{0}{\beta \gamma - \beta y - 1} \}| \\
&=& \ceil{\frac{\beta(\gamma - y)}{\gamma}} \\
&=& \ceil{\beta - \frac{\beta y}{\gamma}} \\
&=& \beta - \floor{\frac{\beta y}{\gamma}}
\end{eqnarray*}
by virtue of Recall \ref{rcll:floor}. 
We thus have:
\begin{eqnarray*}
|A \cap \lr{0}{\beta \gamma - 1}| 
&=& \sum_{y = 0}^{\gamma - 1} \left( \beta - \floor{\frac{\beta y}{\gamma}} \right) \\
&=& \beta \gamma - \sum_{y = 0}^{\gamma - 1} \floor{\frac{\beta y}{\gamma}} \\
&=& \beta \gamma - \frac{(\beta - 1)(\gamma -1)}{2}
\end{eqnarray*}
by virtue of Recall \ref{rcll:sum_floor} since $\beta$ and $\gamma$ are coprime. 
\end{proof}

\begin{lemma}
\label{lmm:symmetric}
Let $o$ be in $\lr{0}{\beta n + \gamma m}$. We have:
\begin{equation}
o \in A \iff \beta n + \gamma m - o \in A
\end{equation}
\end{lemma}

\begin{proof}
Let $o$ be in $\lr{0}{\beta n + \gamma m}$. 
We know that $\beta$ and $\gamma$ are coprime, so we can take two integers $y$ and $z$ in $\Z$ such that:
\[ o = \beta y + \gamma z \]
and we have:
\[ \beta n + \gamma m - o = \beta (n-y) + \gamma (m-z) \]

Now, we have:
\[ (y, z) \in \lr{0}{n} \times \lr{0}{m} \iff (n-y, m-z) \in \lr{0}{n} \times \lr{0}{m} \]
and thus:
\[ o \in A \iff \beta n + \gamma m - o \in A \]
\end{proof}

\begin{corollary}
\label{crllr:count_last}
We have:
\begin{equation}
|A \cap \lr{\beta n + \gamma m - \beta \gamma + 1}{\beta n + \gamma m}| = \beta \gamma - \frac{(\beta - 1)(\gamma -1)}{2}
\end{equation}
\end{corollary}

\begin{proof}
This is an immediate consequence of Theorem \ref{thm:count_first} and Lemma \ref{lmm:symmetric}. 
\end{proof}

\begin{definition}
\label{def:ab}
\begin{enumerate}
\item
For all $i$ in $\lr{0}{n}$, we define:
\begin{equation}
A_i = \{ \beta i + \gamma z \mid z \in \lr{0}{m} \}
\end{equation}
so that $A = \bigcup_{i=0}^{n} A_i$. 

\item
For all $j$ in $\lr{0}{\gamma - 1}$, we define:
\begin{eqnarray*}
B_j &=& \bigcup_{\substack{i \in \lr{0}{n} \\ i = j \mod \gamma}} A_i \\
&=& \bigcup_{k = 0}^{\floor{\frac{n-j}{\gamma}}} A_{j + \gamma k}
\end{eqnarray*}
so that $A$ can be rewritten: 
\[A = \bigcup_{j \in \lr{0}{\gamma - 1}} B_j\]
\end{enumerate}
\end{definition}

\begin{lemma}
\label{lmm:bj}
For all $j$ in $\lr{0}{\gamma - 1}$, we have:
\begin{equation}
B_j = \eqclass{\beta j}{\gamma} \cap \lr{\beta j}{\beta(j + \gamma \floor{\frac{n-j}{\gamma}})  + \gamma m }
\end{equation}
\end{lemma}

\begin{proof}
Let $j$ be in $\lr{0}{\gamma - 1}$. 
For all $q$ in $\lr{0}{\floor{\frac{n-j}{\gamma}}}$, let us define the predicate $P_q$ as follows:
\begin{equation*}
P_q \equiv 
\left(
\bigcup_{k = 0}^{q} A_{j + \gamma k}
= \eqclass{\beta j}{\gamma} \cap \lr{\beta j}{\beta(j + \gamma q)  + \gamma m }
\right)
\end{equation*}
and let us prove by induction that $P_q$ holds for all $q$ in $\lr{0}{\floor{\frac{n-j}{\gamma}}}$. 

By definition, we have:
\[
\bigcup_{k = 0}^{0} A_{j + \gamma k}
= A_j
= \eqclass{\beta j}{\gamma} \cap \lr{\beta j}{\beta j  + \gamma m }
\]
and thus $P_0$ holds. 

Let $q$ be in $\lr{0}{\floor{\frac{n-j}{\gamma}} - 1}$ and let us assume that $P_q$ holds. 

By definition, we have:
\begin{eqnarray*}
A_{j + \gamma(q+1)} &=& \eqclass{\beta (j + \gamma (q+1))}{\gamma} \cap \lr{\beta (j + \gamma(q+1))}{\beta (j + \gamma(q+1))  + \gamma m } \nonumber \\ 
&=& \eqclass{\beta j}{\gamma} \cap \lr{\beta (j + \gamma(q+1))}{\beta (j + \gamma(q+1))  + \gamma m }
\end{eqnarray*}

This equation, combined with the assumption that $P_q$ holds, yields us:
\begin{equation}
\label{eq:contiguous}
\begin{aligned}
\bigcup_{k = 0}^{q + 1} A_{j + \gamma k} = 
\eqclass{\beta j}{\gamma} \cap 
\Big( &
\lr{\beta j}{\beta(j + \gamma q)  + \gamma m } \\
\cup &
\lr{\beta (j + \gamma(q+1))}{\beta (j + \gamma(q+1))  + \gamma m }
\Big)
\end{aligned}
\end{equation}

However, as $m \geq \beta$, we know that:
\[ \beta(j + \gamma q)  + \gamma m \geq \beta (j + \gamma(q+1)) \]

And so Equation (\ref{eq:contiguous}) becomes:
\begin{equation*}
\bigcup_{k = 0}^{q + 1} A_{j + \gamma k} = 
\eqclass{\beta j}{\gamma} \cap 
\lr{\beta j}{\beta (j + \gamma(q+1))  + \gamma m }
\end{equation*}
which means that $P_{q+1}$ holds, which enables us to conclude the induction. 

\end{proof}

\begin{lemma}
\label{lmm:bj_inc}
For all $j$ in $\lr{0}{\gamma - 1}$, we have:
\begin{equation}
\eqclass{\beta j}{\gamma} \cap \lr{\beta \gamma}{\beta n + \gamma m - \beta \gamma}
\subseteq B_j
\end{equation}
\end{lemma}

\begin{proof}
Let $j$ be in $\lr{0}{\gamma - 1}$. 
Lemma \ref{lmm:bj} ensures that we have:
\begin{equation*}
B_j = \eqclass{\beta j}{\gamma} \cap \lr{\beta j}{\beta(j + \gamma \floor{\frac{n-j}{\gamma}})  + \gamma m }
\end{equation*}

First, we know that $\beta j \leq \beta \gamma$. 
Let us now define $S$ as the following statement:
\begin{equation}
\label{eq:boundsup}
S\colon \beta n + \gamma m - \beta \gamma
\leq \beta(j + \gamma \floor{\frac{n-j}{\gamma}})  + \gamma m
\end{equation}

We have the following equivalences:
\begin{eqnarray}
S &\iff& \beta n - \beta \gamma \nonumber
\leq \beta(j + \gamma \floor{\frac{n-j}{\gamma}}) \\
&\iff& n - \gamma
\leq j + \gamma \floor{\frac{n-j}{\gamma}} \nonumber \\
&\iff& n - j
\leq \gamma \floor{\frac{n-j}{\gamma}} + \gamma \nonumber \\
&\iff& \frac{n-j}{\gamma}
\leq \floor{\frac{n-j}{\gamma}} + 1 \label{eq:floortrue}
\end{eqnarray}

By definition of the floor function, Equation (\ref{eq:floortrue}) holds and by equivalence, Equation (\ref{eq:boundsup}) thus also holds, and so:
\[ 
\lr{\beta \gamma}{\beta n + \gamma m - \beta \gamma } 
\subseteq
\lr{\beta j}{\beta(j + \gamma \floor{\frac{n-j}{\gamma}})  + \gamma m } 
\]
and intersecting with $\eqclass{\beta j}{\gamma}$ yields us the expected result. 
\end{proof}

\begin{theorem}
\label{thm:middle_inc}
We have:
\begin{equation}
\lr{\beta \gamma}{\beta n + \gamma m - \beta \gamma} \subseteq A
\end{equation}
\end{theorem}

\begin{proof}
By definition, we have:
\[A = \bigcup_{j \in \lr{0}{\gamma - 1}} B_j\]

Lemma \ref{lmm:bj_inc} thus implies:
\begin{eqnarray*}
A &\supseteq&
\bigcup_{j \in \lr{0}{\gamma - 1}}
\Big(
\eqclass{\beta j}{\gamma} \cap \lr{\beta \gamma}{\beta n + \gamma m - \beta \gamma}
\Big) \\
&\supseteq& 
\left(
\bigcup_{j \in \lr{0}{\gamma - 1}}
\eqclass{\beta j}{\gamma} 
\right) 
\cap \lr{\beta \gamma}{\beta n + \gamma m - \beta \gamma}
\end{eqnarray*}

But as $\beta$ and $\gamma$ are coprime, Recall \ref{rcll:gengroup} ensures that:
\[ \bigcup_{j \in \lr{0}{\gamma - 1}}
\eqclass{\beta j}{\gamma} = \Z \]
which concludes the proof of the Theorem. 
\end{proof}

\begin{theorem}
\label{thm:cardinal}
We have:
\begin{equation}
N_O = |A| = \beta n + \gamma m + 1 - (\beta-1)(\gamma - 1)
\end{equation}
\end{theorem}

\begin{proof}
Let us define the intervals $I_1, I_2$ and $I_3$ as follows:
\begin{eqnarray*}
I_1 &=& \lr{0}{\beta \gamma - 1} \\
I_2 &=& \lr{\beta \gamma}{\beta n + \gamma m - \beta \gamma} \\
I_3 &=& \lr{\beta n + \gamma m - \beta \gamma + 1}{\beta n + \gamma m}
\end{eqnarray*}

We notice that $(I_1, I_2, I_3)$ forms a partition of $\lr{0}{\beta n + \gamma m}$ and thus we can partition $A$ into $(A \cap I_1, A \cap I_2, A \cap I_3)$, which enables us to express the cardinal of $A$ as the following sum:
\begin{equation}
\label{eq:card_sum}
|A| = |A \cap I_1| + |A \cap I_2| + |A \cap I_3|
\end{equation}

We already know by Theorem \ref{thm:count_first} and Corollary \ref{crllr:count_last} that:
\begin{equation*}
|A \cap I_1| = |A \cap I_3| = \beta \gamma - \frac{(\beta - 1)(\gamma -1)}{2}
\end{equation*}

Moreover, Theorem \ref{thm:middle_inc} ensures that:
\[ A \cap I_2 = I_2 \]
and thus:
\begin{eqnarray*}
|A \cap I_2| &=& |I_2| \\
&=& \beta n + \gamma m - 2 \beta \gamma + 1
\end{eqnarray*}
and finally Equation (\ref{eq:card_sum}) becomes:
\begin{eqnarray*}
|A| &=& \beta n + \gamma m - 2 \beta \gamma + 1 + 2 \left(\beta \gamma - \frac{(\beta - 1)(\gamma -1)}{2} \right) \\
&=& \beta n + \gamma m + 1 - (\beta - 1)(\gamma -1)
\end{eqnarray*}
\end{proof}

In the following corollary, we can now synthesise the previous results and realize one of our main objectives: a closed-form expression for $\HH(Y \mid O)$ under uniform prior beliefs. 
\begin{corollary}
\label{crll:h_explicit}
We consider a function $f$ defined as $f(y, z) = \alpha + \beta y + \gamma z$ with $\alpha$, $\beta$ and $\gamma$ being three constant integer values, with non-zero $\beta$ and $\gamma$. 
We assume that $Y$ and $Z$ are ranged in the respective intervals $I_Y$ and $I_Z$ of size $(n+1)$ and $(m+1)$ respectively, and we assume that $Y$ and $Z$ are uniformly distributed on those intervals. 
Let $d$ be the greatest common divisor of $\beta$ and $\gamma$, and let us define $\beta' = | \frac{\beta}{d} |$ and $\gamma' = | \frac{\gamma}{d} |$ where here $|x|$ represents the absolute value of integer $x$. 
Then, we have:
\begin{equation}
\HH(Y \mid O) = - \log \frac{\beta' n + \gamma' m + 1 - (\beta'-1)(\gamma' - 1)}{(n+1)(m+1)} 
\end{equation}
\end{corollary}

\begin{proof}
This is an immediate consequence of Theorem \ref{thm:cardinal} and Equation (\ref{eq:vuln_uni}). 
\end{proof}

This gives us a method for quantifying the information leaks about a targeted party from the public output of an SMC under uniform prior beliefs. This method requires a computational time that is constant in the inputs size and logarithmic in the coefficients of the affine function due to the greatest common divisor operation. 
In the next section, we show how we can reason about $\HH(Y \mid O)$ when an attacker has some prior beliefs about $Y$ and $Z$. 

\section{Privacy bounds under non-uniform prior beliefs}
\label{sec:non_uni}

In this section, we present some lower and upper bounds for $\HH(Y \mid O)$ under non-uniform prior beliefs on the target's and the spectator's input. 
The following theorem first imposes a lower bound.

\begin{theorem}
\label{thm:lower_bound}
We have:
\begin{equation}
\HH(Y\mid O) \geq \HH(Y) + \HH(Z) - \log(N_O)
\end{equation}
with equality when $Y$ and $Z$ have uniform prior distributions. 
\end{theorem}

\begin{proof}
We denote by $N_O$ the number of possible outputs, for which we recall that Theorem \ref{thm:cardinal} yields an explicit formula. 
By comparison between the $1$-norm and the infinity-norm, we have:
\begin{eqnarray*}
\V(Y \mid O) &=& \sum_o \max_y p(y) \cdot p(o \mid y) \\
&\leq& N_O \max_o \max_y p(y) \cdot p(o \mid y) \\
&\leq& N_O \max_y p(y) \cdot \left( \max_o p(o \mid y) \right)
\end{eqnarray*}
But we know that:
\[ \forall y \in I_Y \colon \max_o p(o \mid y) = \max_z p(z) \]
and thus:
\begin{eqnarray*}
\V(Y \mid O) &\leq& N_O \max_y p(y) \cdot \max_z p(z)
\end{eqnarray*}
Taking the negative logarithm concludes the proof, and we can verify, with our bespoke explicit formula from Equation (\ref{eq:vuln_uni}), that we indeed have equality when $Y$ and $Z$ are uniformly distributed since then $\max_y p(y) = \frac{1}{|I_Y|}$ and $\max_z p(z) = \frac{1}{|I_Z|}$. 
\end{proof}

We will now study some upper bounds for $\HH(Y \mid O)$. It is a known result on the min-entropy that $\HH(Y \mid O) \leq \HH(Y)$, i.e. that knowledge of the public output cannot increase the targeted input's entropy. 
We will now prove that $\HH(Y \mid O) \leq \HH(Z)$, i.e. the remaining entropy of $Y$ given knowledge of $O$ cannot be larger than the prior entropy of the spectator's input $Z$. 
To this end, we first state in the next theorem that an attacker eavesdropping on the value of $x$ and learning the public output will gain the same amount of information on targeted input $Y$ than on the spectator's input $Z$. 

\begin{theorem}
\label{thm:equalityYZ}
We have:
\begin{equation}
\HH(Y \mid O) = \HH(Z \mid O)
\end{equation}
\end{theorem}

\begin{proof}
We have:
\[ \HH(Y \mid O) = \sum_o \max_y p(y) p(o \mid y) \]

For all $o$ in $D_O$, we define the set $S^o$ of pairs that result in output $o$ as:
\[ S^o = \{ (y, z) \in I_Y \times I_Z \mid \beta y + \gamma z = o \} \]

We define its projections $S^o_Y$ and $S^o_Z$ on its first and second components respectively as follows:
\begin{eqnarray*}
S^o_Y &=& \{ y \in I_Y \mid \exists z \in I_Z \colon (y, z) \in S^o \} \\
S^o_Z &=& \{ z \in I_Z \mid \exists y \in I_Y \colon (y, z) \in S^o \}
\end{eqnarray*}

Moreover, for all $o$ in $D_O$ and $y$ in $D_Y$, we know that $p(y \mid o)$ is non zero only if there exists a $z$ in $D_Z$ such that $(y, z)$ is in $S_o$. Thus:

\[ \HH(Y \mid O) = \sum_o \max_{y \in S^o_Y} p(y) p(o \mid y) \]

Now, for all $o$ in $D_O$ and $y$ in $S^o_Y$, there exists a unique $z$ in $D_Z$ such that $\beta y + \gamma z = o$, determined by $z = \frac{o - \beta y}{\gamma}$ and we have $p(o \mid y) = p(z)$. Thus:

\[ \HH(Y \mid O) = \sum_o \max_{(y, z) \in S^o} p(y) p(z) \]

Conversely, for all $o$ in $D_O$ and $z$ in $S^o_Z$, there exists a unique $y$ satisfying $\beta y + \gamma z = o$ and we have $p(o \mid z)$. Thus: 

\[ \HH(Y \mid O) = \sum_o \max_{z \in S^o_Z} p(o \mid z) p(z) \]

And finally for all $o$ in $D_O$, we know that $p(o \mid z)$ can only be non zero if $z$ is in $S^o_Z$ and thus:
\begin{eqnarray*}
\HH(Y \mid O) &=& \sum_o \max_{z \in D_Z} p(o \mid z) p(z) \\
&=& \HH(Z \mid O)
\end{eqnarray*}
\end{proof}

\begin{corollary}
\label{crll:lower_z}
We have:
\[ \HH(Y \mid O) \leq \HH(Z) \]
\end{corollary}

\begin{proof}
We have $\HH(Z\mid O) \leq \HH(Z)$ and Theorem \ref{thm:equalityYZ} concludes the proof. 
\end{proof}

\section{Examples}
\label{sec:examples}

In this section, we illustrate the theoretical results previously obtained. 
We begin this section by presenting an example which deepens our understanding of the behaviour of $\HH(Y \mid O)$ under non-uniform prior beliefs. 
We could intuitively posit that $\HH(Y \mid O)$ is maximal when the prior distributions for $Y$ and $Z$ are uniform. However, we refute this hypothesis in the following example. 

\begin{example}
Let us consider the function $f(y, z) = o = 2y + 3z$, and let us assume that $y$ in ranged in $I_Y = \lr{0}{2}$ and $z$ is ranged in $I_Z = \lr{0}{1}$. 
Let us consider $\pi^u_Y = \{ 0:\frac{1}{3}, 1:\frac{1}{3}, 2:\frac{1}{3} \}$ and $\pi^u_Z = \{ 0:\frac{1}{2}, 1:\frac{1}{2} \}$ the uniform distributions for $Y$ and $Z$ on their respective domains. We also consider the following particular distribution $\pi_Y^* = \{ 0:\frac{1}{2}, 1:0, 2:\frac{1}{2} \}$. 

Then, when $Y$ and $Z$ respectively follow the prior distributions $\pi_Y^u$ and $\pi_Z^u$, we have $\HH^u(Y \mid O) = -\log(\frac{5}{6})$. 
On the other hand, when $Y$ and $Z$ respectively follow the prior distributions $\pi_Y^*$ and $\pi_Z^u$, we get $\HH^*(Y \mid O) = -\log(\frac{3}{4})$. 

We thus have $\HH^*(Y \mid O) > \HH^u(Y \mid O)$ which contradicts the intuitive hypothesis. 
\end{example}

The next example presents a use case of Corollary \ref{crll:h_explicit}. 

\begin{example}
\label{ex:uni}
Let us consider three parties $\cX$, $\cY$ and $\cZ$ holding respective private inputs $x$, $y$ and $z$ and willing to enter the secure computation of a public function $f$ defined as $f(x, y, z) = (3x-6)y + (x^2-2x+6)z$. 
We suppose that party $\cX$ is attacking input $y$ under spectator $\cZ$. We notice that when input $x$ is fixed, the function $f$ is affine in $y$ and $z$ and we can thus apply our privacy analysis. 
We assume that $Y$ and $Z$ are ranged in the input domain $I = \lr{0}{\pow{5}{12}}$ and we assume that $\cX$'s prior beliefs $\pY$ and $\pZ$ on those inputs are uniform over $I$. 
We plot in Figure \ref{fig:ex_uni} the values of $\HH(Y \mid O)$ computed via Corollary \ref{crll:h_explicit}, for the values of input $x$ ranged in $\lr{0}{30}$. 
Note that although a small interval for the values of $x$ has been chosen for readability purposes, entropy $\HH(Y \mid O)$ can be computed for any value of $x$. 
For an attacker $\cX$ who is willing to lie on his honest and intended input in order to learn as much information as possible on his targeted input $Y$, he would have more incentive to enter some value $x$ that produces low entropy. For example, he would rather enter value $x=25$ than $x=2$. 
Conversely, targeted parties could consider such information so as to evaluate the risk that they would face by entering the computation in the worst case, or on average. 

\begin{figure}
\centering
\begin{tikzpicture}[scale=.9]
	\begin{axis}[
	scale=1,
	 ymin=32,
	 xmin=0,
	 xmax=30,
	  xlabel=$x$,
	  ylabel=entropy,
	  legend pos=north east]
	\addplot[blue] table [y=fx, x=x, mark=none]{uni.dat};
	\addlegendentry{$\HH(Y \mid O)$}
\end{axis}
\end{tikzpicture}
\caption{Illustration of $\HH(Y \mid O)$ in the computation of function $f(x, y, z) = (3x-6)y + (x^2-2x+6)z$ with uniform prior beliefs $\pY$ and $\pZ$ over $\lr{0}{\pow{5}{12}}$, when $X$ attacks $Y$ under spectator $Z$. 
}
\label{fig:ex_uni}
\end{figure}
\end{example}

\begin{example}
In order to evaluate the effectiveness of our approach, we repeated the operations of the previous example while letting the size of the input spaces $I_Y$ and $I_Z$ vary, and by comparing the computational time that different methods require to perform such analyses. 
More precisely, we computed the $31$ values of $\HH(Y \mid O)$ in the same scenario as in the previous Example \ref{ex:uni}, but we let inputs $Y$ and $Z$ be ranged in the intervals $\lr{0}{\pow{5}{p}}$ for different values of $p$. 
We compared the time taken by the three following methods, which we display in Figure \ref{fig:times}.  
\begin{enumerate}
\item \emph{Naive method:} We use the combinatorial formula given in Equation (\ref{eq:awae_naive}) that has complexity $\mathcal{O}(n^2m^2)$. 
\item \emph{Simplified method:} We use the simplified formula of Equation (\ref{eq:vuln_uni}) for affine functions under uniform distributions, where $N_O$ is computed naively by enumerating the set of outputs, which yields complexity $\mathcal{O}(nm)$. 
\item \emph{Explicit method:} We use the result of Corollary \ref{crll:h_explicit} which provides a constant time formula. 
\end{enumerate}

\begin{figure}
\centering
\[
\begin{array}{l|lll}
   p & \text{Naive} & \text{Simplified} & \text{Explicit} \\ \hline
   0 &  \pow{9.1}{-3} & \pow{8.9}{-4} & \pow{1.5}{-4} \\
   1 & \pow{3.7}{-1} & \pow{1.7}{-2} & \pow{9.4}{-5} \\
   2 & 32 & 1.9 & \pow{1.2}{-4} \\
   3 & \infty & 216 & \pow{9.9}{-5} \\
   4 & \infty & \infty & \pow{1.2}{-4} \\
   12 & \infty & \infty & \pow{1.2}{-4} \\
\end{array}
\]
\caption{Computational time (in seconds) taken by our explicit formula compared to naive methods in the privacy analysis of Example \ref{ex:uni} with inputs ranged in $\lr{0}{\pow{5}{p}}$ for different values of $p$. }
\label{fig:times}
\end{figure}

The variables $n$ and $m$ represent the size of the input spaces (minus $1$) and are both set to $\pow{5}{p}$ for varying values of $p$. 
We set a time limit of $5$ minutes and we mark by an infinity sign the computations that timed out. 
We can notice that both naive methods rapidly time out as the input space grows whereas our explicit formula enables us to perform privacy analyses in constant time for arbitrarily large input spaces such as the one performed in Example \ref{ex:uni}. 
Those computations have been performed on an Intel(R) Core(TM) i3-2350M CPU @ 2.30GHz, but are aimed at estimating the order of magnitude of those methods rather than precisely assessing them individually. 

\end{example}

In the following example, we now illustrate the lower and upper bounds that have been derived for $\HH(Y \mid O)$ under non-uniform prior beliefs. 

\begin{example}
We consider the computation of $f$ whose simplification, once $x$ is fixed, is defined as $f(y, z) = y + z$. 
We assume that $Y$ and $Z$ are ranged in the domain $I=\lr{0}{50}$. 

We define a spiked distribution $d_s$ parametrised by a domain $D$, a center $c \in D$ and a weight $w \in [0, 1]$ as:
\begin{eqnarray*}
d_s(D, c, w)[c] &=& w \\
d_s(D, c, w)[x] &=& \frac{1 - w}{|D| - 1} \qquad \text{ if $x \neq c$}
\end{eqnarray*}
In other words, distribution $d_s(D, c, w)$ allocates a probability $w$ to the value $c$ and distributes the remaining probability uniformly amongst the other values of the domain $D$. 

We suppose that $Y$ is uniformly distributed over $I$ and that $Z$ follows distribution $d_s(I, 0, w)$ for different weights $w$. 
We divide the interval $[0,1]$ into $50$ values. For each $w$ in those $50$ values, we compute the exact values of $\HH(Y \mid O)$ and that of its bounds derived in the previous section. 
The value of $\HH(Y \mid O)$ appears in blue in Figure \ref{fig:ex_non_uni}. Its lower bound stemming from Theorem \ref{thm:lower_bound} is drawn in red and its upper bound derived from Corollary \ref{crll:lower_z} is traced in green. 
Note that we considered small input spaces since $\HH(Y \mid O)$ is here calculated with a naive method, although its bounds can be computed efficiently for arbitrarily large input spaces. 

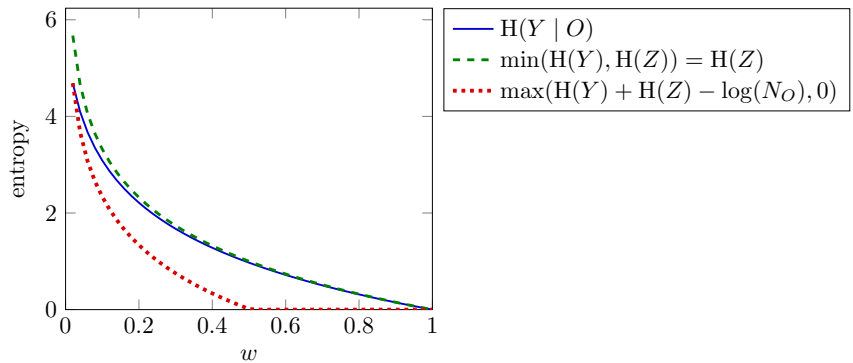
\begin{figure}
\centering
\begin{tikzpicture}[scale=.9]
	\begin{axis}[
	scale=1,
	 ymin=0,
	 xmin=0,
	 xmax=1,
	  xlabel=$w$,
	  ylabel=entropy,
	  legend pos=outer north east,
	  legend cell align={left}]
	\addplot[blue!80!black, thick] table [y=fx, x=x, mark=none]{non_uni.dat};
	\addlegendentry{$\HH(Y \mid O)$}
	\addplot[green!50!black, dashed, very thick] table [y=upper, x=x, mark=none]{non_uni.dat};
	\addlegendentry{$\min(\HH(Y),\HH(Z)) = \HH(Z)$}
	\addplot[red!85!black, dotted, ultra thick] table [y=lower, x=x, mark=none]{non_uni.dat};
	\addlegendentry{$\max(\HH(Y) + \HH(Z) - \log(N_O),0)$}
\end{axis}
\end{tikzpicture}
\caption{Behaviour of $\HH(Y \mid O)$ in the computation of function $f(y, z) = y + z$ for $Y$ and $Z$ ranged in $I=\lr{0}{50}$, with $\pZ = d_s(I, 0, w)$ and $\pY$ uniform over $I$. 
}
\label{fig:ex_non_uni}
\end{figure}

\end{example}

\section{Conclusion}
\label{sec:conclu}

Although extensive researches in Secure Multi-party Computation have considerably improved the efficiency of cryptographic protocols, the quantification of the acceptable leakage is a problem that still requires deeper investigations. 
Indeed, the computational complexity of those recently introduced privacy analyses does not yet allow their application in practical situations that involve large input spaces. 
In this work, we focused our attention on secure three-party computations of affine functions. 
We have formally investigated the behaviour of the acceptable leakage under uniform prior beliefs in order to obtain an explicit formula for the min-entropy of the targeted input given conditional knowledge of the output. 
The calculation of this closed-form expression requires a computational time that is constant in the inputs sizes and logarithmic in the coefficients of the function, which enables the privacy analysis of such computations in practice. 
Finally, we have derived some theoretical bounds for this acceptable leakage when the input prior distributions are non-uniform in order to accommodate the potential prior belief that an attacker may have. 

In the future, we would like to enlarge our understanding of the acceptable leakage in more general settings. 
First, as our work is motivated by the privacy leaks that occur during SMC protocols, we tailored our analyses for finite input spaces. However, it would be interesting to adapt our model and to design some methods that can accommodate continuous input and output spaces. 
Moreover, although our current analysis considers the computation of affine functions for three parties, it would be of interest to explore the computation of affine functions for any number of parties. 

We also mean to investigate more general functions that involve non-linear terms. It would be particularly interesting to study the composition of our analyses of affine functions in order to use them as building blocks for studying more complex functions. 
Finally, efficient and exact quantification of the acceptable leakage for general functions may be hard to obtain simultaneously, and we would thus also be interested in providing efficient methods for approximating the inputs privacy in general scenarios. 

%\textbf{Acknowledgements:} 
%This research was supported by the UK EPSRC grants EP/N020030/1 and EP/N023242/1. 

\addcontentsline{toc}{section}{References}
\bibliographystyle{unsrt}%Used BibTeX style is unsrt
\bibliography{./references}

\end{document}